\documentclass[a4paper,english,11pt]{amsart} %draft: vede i fuori misura
\usepackage{amsmath}
\usepackage{amsthm}
\usepackage{amsfonts}
\usepackage{amssymb}
\usepackage{graphicx}
\usepackage{xcolor}
\usepackage{enumerate}
\usepackage[english]{babel}
\usepackage{tikz}

\pdfoutput=1

\theoremstyle{definition}
\newtheorem{definition}{Definition}

\theoremstyle{remark}
\newtheorem{remark}[definition]{Remark}

\theoremstyle{plain}
\newtheorem{theorem}[definition]{Theorem}
\newtheorem{lma}[definition]{Lemma}

\newtheorem{prop}[definition]{Proposition}

\newtheoremstyle{mystyle}
{}
{}
{}
{}
{\bfseries}
{}
{\newline}
{}

\DeclareMathOperator{\Sym}{Sym}

\DeclareMathOperator{\Pol}{Pol}

\def\A{\mathcal{A}}
\def\aux{\textup{aux}}
\def\BV{\mathrm{BV}}
\def\C{\mathbb{C}}
\def\diag{\textup{diag}}
\def\E{\mathcal{E}}
\def\ferm{\textup{ferm}}
\def\H{\mathcal{H}}
\def\off{\textup{off}}
\def\phi{\varphi}
\def\R{\mathbb{R}}
\def\tot{\textup{tot}}
\DeclareMathOperator{\tr}{Tr}
\def\Z{\mathbb{Z}}

\makeindex

\title[NCG and the BV formalism]{Noncommutative Geometry and the BV formalism: application to a matrix model}

\author{Roberta A. Iseppi}
\address{Max Planck Institute for Mathematics, Vivatsgasse 7, 53111 Bonn, Germany}
\email{iseppi@mpim-bonn.mpg.de}

\author{Walter D. van Suijlekom}
\address{Institute for Mathematics, Astrophysics and Particle Physics, Radboud University Nijmegen, Heyendaalseweg 135, 6525 AJ Nijmegen, The Netherlands}
\email{waltervs@math.ru.nl} 

\date{\today}

\begin{document}
\maketitle

\begin{abstract}
We analyze a $U(2)$-matrix model derived from a finite spectral triple. By applying the BV formalism, we find a general solution to the classical master equation. To describe the BV formalism in the context of noncommutative geometry, we define two finite spectral triples: the BV spectral triple and the BV auxiliary spectral triple. These are constructed from the gauge fields, ghost fields and anti-fields that enter the BV construction. We show that their fermionic actions add up precisely to the BV action. This approach allows for a geometric description of the ghost fields and their properties in terms of the BV spectral triple. 
\end{abstract}

\section{Introduction}
\label{Section: Introduction}
\noindent
Since the early days of noncommutative geometry \cite{Connes} it has been clear that this mathematical theory is strongly related to gauge theories in physics. Indeed, gauge theories are naturally induced by spectral triples, where the noncommutativity of the pertinent algebra naturally gives rise to non-abelian gauge groups. This has successfully been applied to Yang--Mills gauge theories \cite{CC96} and to the celebrated Standard Model of particle physics \cite{CCM07}. It is also clear that in the finite-dimensional case, when the algebras are matrix algebras, one obtains hermitian matrix models. 

A powerful method to analyze the nature of the gauge symmetries in gauge theories ---with the eventual purpose of understanding their rigorous quantization--- is the BRST formalism \cite{BRS,BRS3,T} and its far-reaching extension, the BV formalism \cite{BV1,BV2} ({\it cf.} \cite{Fior,GPS,Schw} for review articles). A first key ingredient in both of these formalisms are Faddeev-Popov {\it ghost fields} \cite{Faddev-Popov}, which are introduced to cancel the physically irrelevant gauge symmetries. The BV formalism then proceeds by introducing also so-called {\it anti-fields} for all previously defined gauge and ghost fields. Moreover, an extended action functional is defined as a solution to the so-called `classical master equation' ({\it cf.} Definition \ref{defn:bv} below).

\bigskip

We start this paper by recalling ({\it cf.} \cite{primo_articolo}) the result obtained by applying the BV formalism to a $U(2)$-matrix model, which is derived from a finite spectral triple on the algebra $M_2(\C)$. We find that the gauge structure of this model is richer than expected, requiring also the introduction of ghost-for-ghost fields. After having added the necessary anti-fields, we state the general form of the extended action that solves the classical master equation. Then, the construction is finished by determining the BV auxiliary pairs, which are essential in order to perform a gauge-fixing procedure. As such, our constructions fits nicely with previous studies of the BV formalism applied to gauge models derived from noncommutative geometry, such as \cite{beringgrosse,hanlonjoshi,huffel2001,huffel2002}. 

As a next step, we here define two spectral triples ---the BV spectral triple and the BV auxiliary spectral triple--- for which the fermionic action functionals sum up precisely to the so-called BV action functional, which is defined to be the difference between the extended action functional and the initial action. Thus we obtain a noncommutative geometric description of the BV formalism for this particular model, which by itself was derived from a spectral triple.

With this model we give the first description of the BV formalism completely in terms of noncommutative geometric data, that is to say, spectral triples. It serves as a guiding example for higher-rank, $U(n)$-matrix models and eventually for physically realistic gauge theories defined on a manifold. However, an analysis of these models goes beyond the scope of this paper and is left for future research. 

The paper is organized as follows. In Section \ref{Section: The noncommutative geometry setting} we quickly review the notion of a spectral triple and explain how gauge theories derive from it. Section \ref{BV sect} contains a concise overview of the BV formalism, geared towards our finite-dimensional case and essentially following \cite{felder} (see also \cite{primo_articolo}). 

In Section \ref{algorithm applied to matrix model} we recall what we obtained by applying the BV formalism to a $U(2)$-matrix model, understood as a gauge theory that is obtained from a spectral triple. Section \ref{chapter: NCG and the BV approach} is the heart of this paper: we construct a so-called BV spectral triple and BV auxiliary spectral triple and show that the sum of the corresponding fermionic actions coincides with the BV action functional. 

\subsection*{Acknowledgements}
The research presented in this article was partially supported by the Netherlands Organization for Scientific Research (NWO), through Vrije Competitie project number 613.000.910. The first author would like to thank the Max Planck Institute for Mathematics in Bonn, where the final stages of writing this article took place. Moreover, she would like to thank Matilde Marcolli for interesting comments and inspiring discussions and Klaas Landsman for useful remarks.

\section{Preliminaries}
\subsection{The noncommutative geometry setting}
\label{Section: The noncommutative geometry setting}
We recall the notion of a {\em spectral triple} and the construction of the canonically induced gauge theory 
 ({\it cf.}  \cite{C96}, \cite[Sect. 1.10]{marcolli} and \cite[Ch. 6]{walter}). 
This method will be later applied to a finite spectral triple that yields a $U(2)$-gauge theory, which we want to analyze using the BV formalism. 

  \begin{definition}
  \label{def spectral triple}
   A {\em spectral triple} $(\mathcal{A}, \mathcal{H}, D)$  consists of an involutive unital algebra $\mathcal{A}$, faithfully represented as operators on a Hilbert space $\mathcal{H}$, together with a self-adjoint operator $D$ on $\mathcal{H}$, with a compact resolvent,  such that 
   the commutators $[D, a]$ are bounded operators for each $a \in A$.  
  \end{definition}

\begin{remark}
The spectral triple $(\mathcal{A}, \mathcal{H}, D)$ is said to be {\em finite} if $\mathcal{H}$ is finite dimensional. 
By a classical result the algebra $\mathcal{A}$ in this case has to be a direct sum of matrix algebras, i.e.
$$ \mathcal{A}\simeq \bigoplus_{i=1}^{k}M_{n_{i}}(\mathbb{C})
 %,  \quad \quad\quad n_{1}, \dots, n_{k} \in \mathbb{N}.
 $$
 for positive integers $n_{1}, \dots, n_{k}$. 
Moreover, the required conditions on the self-adjoint operator $D$ are automatically satisfied in this finite-dimensional setting.
\end{remark}
   
  \begin{definition}
   An {\em  even spectral triple} $(\mathcal{A}, \mathcal{H}, D)$ is one in which the Hilbert space $\mathcal{H}$ is endowed with a $\mathbb{Z}/2$-grading $\gamma$, given by a linear map $\gamma: \mathcal{H} \to \mathcal{H}$, such that 
$$   D\gamma = -\gamma D  \quad \text{and} \quad   \gamma a = a \gamma 
$$
for all $a \in \mathcal{A}$.
  \end{definition}
  
\begin{definition}
\label{def real structure}
A {\em real structure of KO-dimension $n$ (mod 8)} on a spectral triple $(\mathcal{A}, \mathcal{H}, D)$ is  an anti-linear isometry $ J: \mathcal{H}\rightarrow \mathcal{H}$ that satisfies 
$$
J^2 = \epsilon  \quad \text{and} \quad  J D = \epsilon^{\prime} DJ
$$
together with the condition $$ J \gamma = \epsilon^{\prime \prime} \gamma J $$ in the even case. 
%$$J^2 = \epsilon, \quad \quad \quad \quad J D = \epsilon^{\prime} DJ,  \quad \quad \quad \quad J \gamma = \epsilon^{\prime \prime} \gamma J, \ \mbox{(even case)}.$$
The constants $\epsilon$, $\epsilon^{\prime}$ and $\epsilon^{\prime \prime}$  depend on the KO-dimension $n$ (mod 8) as follows:
$$
\begin{array}{|c|cccccccc|}
 \hline
 {n} & 0 & 1 & 2 & 3 & 4 & 5 & 6 & 7\\
 \hline
 \epsilon & 1 & \phantom{-}1 & -1 & -1 & -1 & -1 & \phantom{-}1 & 1\\
 \epsilon^{\prime} & 1 & -1 & \phantom{-} 1 & \phantom{-}1 & \phantom{-}1 & -1 & \phantom{-}1 & 1 \\ 
 \epsilon^{\prime \prime} & 1 & & -1 & & \phantom{-}1 & & -1 &\\
 \hline
\end{array}
$$
Moreover, we require for all $a,b\in \A$ that:
\begin{itemize}
\item[-] the action of $\mathcal{A}$ satisfies the {\em commutation rule}: $\big[a,  J b^{*}J^{-1}\big]=0$;
\item[-] the operator $D$ fulfills the {\em first-order condition}: $[ [ D, a ], J b^{*} J^{-1}] = 0$.
\end{itemize} 
When a spectral triple $(\mathcal{A}, \mathcal{H}, D)$ is endowed with such a real structure $J$, it is said to be a {\em real spectral triple} and denoted by $(\mathcal{A}, \mathcal{H}, D, J)$.
 \end{definition}

 Given a possibly real spectral triple, there are two notions of action functionals related to it: the {\em spectral action} and the {\em fermionic action} ({\it cf.} \cite{CC97,CC96,CCM07}).
  
\begin{definition}
\label{def spectral action}
For a finite spectral triple  $(\mathcal{A}, \mathcal{H}, D)$  and a suitable real-valued function $f$, the {\em spectral action} $S_{0}$ is given by
$$S_{0} [D+M]: = \tr \big( f(D + M)\big)
$$
with, as domain, the set of self-adjoint operators of the form $M = \sum_{j} a_{j} [D, b_{j}]$, for $a_{j}, b_{j} \in \mathcal{A}$.
\end{definition}

\begin{remark}
In the finite-dimensional setting, a family of suitable functions $f$ is given by the polynomials in $\Pol_\R(x)$. 
 \end{remark}

\begin{definition}
For  a   finite   spectral triple $(\mathcal{A}, \mathcal{H}, D)$  (finite {real} spectral triple $(\mathcal{A}, \mathcal{H}, D, J)$) 
%and a given Hilbert subspace $\mathcal{H}^{\prime} \subseteq \mathcal{H}$, 
the {\em fermionic action}  on $\mathcal{H}$ is given by
$$S_{\ferm}[\varphi] 
= \frac{1}{2} \langle \varphi , D \varphi \rangle 
\quad \Big(S_{\ferm}[\varphi] = \frac{1}{2} \langle J\varphi , D \varphi \rangle  \Big) ; \qquad (\varphi \in \H).
$$% for all $ \varphi \in \mathcal{H}$.
 \end{definition}

%Let $(\mathcal{A}, \mathcal{H}, D, (J))$ be a finite (real) spectral triple and fix a Hilbert subspace $\mathcal{H}^{\prime} \subseteq \mathcal{H}$. Then the {\em fermionic action} is defined on this subspace by:
%\begin{equation}
% \label{fermionic action}
%S_{\ferm}[\varphi] = \frac{1}{2} \langle (J)\varphi , D \varphi \rangle, \quad \quad \varphi \in \mathcal{H}^{\prime}.
% \end{equation}
% \end{definition}
%% \begin{remark}
%% In the real case, 
%% the choice of the subspace $\mathcal{H}^{\prime} \subseteq \mathcal{H}$ for a fermionic action also depends on the KO-dimension. Indeed, for a spectral triple of KO-dimension $2$ (mod $8$), the subspace $\mathcal{H}^{\prime}$ is taken to be the even part of $\mathcal{H}$, i.e.
%%  $$ \mathcal{H}^{\prime} = \mathcal{H}^{+}:= \big\{ \varphi \in \mathcal{H}, \ \gamma(\varphi) = \varphi \big\}$$
%%  together with the assumption that its elements are classical fermions (i.e. Grassmannian variables).
%% However, for other KO-dimensions, different subspaces $\mathcal{H}^{\prime}$ are possible (cf. Sect. \ref{chapter: NCG and the BV approach} below).
%% \end{remark}

\subsubsection{Gauge theories from spectral triples}
We recall the construction of the gauge theory naturally induced by a spectral triple, restricting to the finite-dimensional case. In this context, the appropriate  notion of a  gauge theory is as follows. 

\begin{definition}
For a real  vector space $X_{0}$ and a  real-valued functional  $S_{0}$  on $X_{0}$,  let 
$F\colon \mathcal{G} \times X_{0} \rightarrow X_{0}$ be a group action on $X_{0}$ for a given group $\mathcal{G}$.
Then the pair $(X_{0}, S_{0})$ is called a {\em gauge theory with gauge group $\mathcal{G}$} if 
 $$S_{0}(F(g, M)) = S_{0}(M)$$
for all $M \in X_{0}$ and $g \in \mathcal{G}$. The space $X_{0}$ is  referred to as the {\em configuration space}, an element $M \in X_{0}$ is called a {\em gauge field} and $S_{0}$ is the {\em action functional}. 
\end{definition}

Given this definition, the derived gauge theory for a finite spectral triple is obtained by the following standard result. 

\begin{prop}
\label{Prop: gauge theory from spectral triple}
For a  finite spectral triple $(\mathcal{A}, \mathcal{H}, D)$, let  
 $$X_{0} = \Big\{ M = \sum_{j} a_{j} \big[ D, b_{j}\big]:  M^{*} = M, \ a_{j}, b_{j} \in \mathcal{A}  \Big\} $$
 be the space of inner fluctuations, the group $\mathcal{G} $ be the unitary elements $\mathcal{U}(\mathcal{A}) = \{ u \in \mathcal{A}: u u^{*}= u^{*}u =1 \}$ of $\mathcal{A}$ acting on $X_0$ via the map 
$
%F:\mathcal{G}  \times  X_{0}  \to X_{0},\qquad
 (u, M) \mapsto  u M u^{*} + u[D, u^{*}] 
$,
and $S_0$ be the spectral action $$S_{0} [M]: = \tr \big( f(D + M)\big) ,$$ for any $M
\in X_0$ and some $f\in \Pol_{\mathbb{R}}(x)$.
Then the pair $(X_{0}, S_{0})$ is a gauge theory with gauge group $\mathcal{G}$.
%
%Given a finite spectral triple $(\mathcal{A}, \mathcal{H}, D)$ we set: 
%\begin{itemize}
% \item $X_{0}: = \big\{ M = \sum_{j} a_{j} \big[ D, b_{j}\big]:  M^{*} = M, \ a_{j}, b_{j} \in \mathcal{A}  \big\}$;
% \item $\mathcal{G} := \mathcal{U}(\mathcal{A}) = \big\{ u \in \mathcal{A}: u u^{*}= u^{*}u =1 \big\}$, acting on $X_0$ via the map 
%\begin{align*}
%F:\mathcal{G}  \times  X_{0}  \to X_{0},\qquad
% (u, M) \mapsto  u M u^{*} + u[D, u^{*}];
%\end{align*}
%\item $S_{0} [M]: = \tr \big( f(D + M)\big) $ with $f\in \Pol_{\mathbb{R}}(x)$.
%\end{itemize}
%Then the pair $(X_{0}, S_{0})$ is a gauge theory with gauge group $\mathcal{G}$.
\end{prop}

\subsection{The BV approach to gauge theories}
\label{BV sect}
As already mentioned, starting with a gauge theory $(X_0,S_0)$, the BV construction is a procedure to determine a corresponding {\em extended theory} $(\widetilde{X}, \widetilde{S})$ via the introduction of {\em ghost/anti-ghost fields.} Here we outline the main aspects of the BV formalism, referring to \cite{felder,GPS,primo_articolo} and references therein for a more exhaustive presentation. 

For notational purposes, it is convenient to fix a basis for $X_0$ so that a gauge field $M \in X_0$ can be written as a vector $M =  (M_a)$, with $a =1, \ldots,n=\dim X_0$ and 
$$X_0 \simeq \langle M_1, \ldots, M_n \rangle_\R.$$
The presence of gauge symmetries in the action demands for the introduction of {\it ghost fields}. In order to determine the number of required ghost fields, one considers the relations $R_i^a$ ($i=1, \ldots,m_0$) between the partial derivatives $\partial_a S_0$ of the action functional with respect to $M_a$, i.e.
$$
(\partial_a S_0)R_i^a = 0.
$$
These relations $R_i^a$ are considered in $\mathcal{O}_{X_{0}}$, which is the ring of regular functions on $X_{0}$. Given each relation $R_i^a$ we introduce a {\it ghost field} $C_i$ for $i=1,\ldots, m_0$. 
It is useful to assign, for good book-keeping, a {\it ghost degree} $\deg(\phi) \in \mathbb Z$ and parity $\epsilon(\phi) \in \{ 0,1\}$ to the fields $\phi$ obtained so far, with $\epsilon(\phi) := \deg(\phi) (\mbox{mod } 2)$. The parity indicates whether the field is a real variable ($\epsilon=0$) or a Grassmannian, namely anti-commuting, variable ($\epsilon =1$). Naturally, we assign 
$$
\deg(M_a) = 0 , \qquad \deg(C_i)=1.
$$
However, it might happen that there are additional relations between the $R_i^a$  themselves. If this happens, the gauge theory is called \emph{reducible} and one has to add ghost-for-ghost fields, denoted by $E_j$, for each such relation-between-relations that appears. The ghost degree of $E_j$ is now $2$. This might continue to ghosts-for-ghosts-for-ghosts  all the way up to the `level of reducibility' $L$, which is the highest appearing ghost degree minus $1$. We refer {\it e.g.} to \cite{GPS} for full details. We denote the resulting configuration space as follows:
$$
\E := \langle M_1, \ldots, M_n \rangle_0 \oplus \langle C_1,\ldots, C_{m_0} \rangle_1 \oplus \langle E_1,\ldots, E_{m_1} \rangle_2 \oplus \cdots
$$ 
The key point to the BV formalism is the introduction of anti-fields for all previously introduced gauge fields, ghost fields, ghost-for-ghost fields, {\it et cetera}. For $\phi \in \E$ we denote the corresponding anti-field by $\phi^*$ and assign ghost degree:
$$
\text{deg}(\varphi^{*}) = - \text{deg}(\varphi) -1.
$$
This results in the vector space
$$
\E^*[1] := \cdots \oplus  \langle E_1^*,\ldots, E_{m_1}^* \rangle_{-3} \oplus \langle C_1^*,\ldots, C_{m_0}^* \rangle_{-2} \oplus \langle M_1^*, \ldots, M_n^* \rangle_{-1}
$$
which is modelled on the dual space $\E^*$,   where  the notation $[1]$ indicates the shift of degree by one, that is to say, %In addition, given a graded module $\mathcal{E}$, we denote by $\mathcal{E}^{*}[1]$ the corresponding dual module, shifted by $1$:
 $$\mathcal{E}^{*}[1] = \bigoplus_{i \in \mathbb{Z}}\big[\mathcal{E}^{*}[1]\big]^{i} \quad \quad \mbox{ with } \quad \quad \big[\mathcal{E}^{*}[1]\big]^{i} = [\mathcal{E}^{*}]^{i+1}.$$ 
The fields and anti-fields are combined into an \emph{extended configuration space} 
\begin{equation} \label{eqn:Xconfsp}
\widetilde X := \E \oplus \E^*[1],
\end{equation}
which has the structure of a super $\Z$-graded vector space. 
In view of this construction, the space of functionals on $\widetilde{X}$ is described by the {\em algebra  $\mathcal{O}_{\widetilde{X}}$ of regular functions on $\widetilde{X}$}, which is the symmetric algebra generated by the $\mathbb{Z}$-graded $\mathcal{O}_{X_{0}}$-module  $\widetilde{X}$ over the ring $\mathcal{O}_{X_{0}}$:
$$\mathcal{O}_{\widetilde{X}}=\Sym_{\mathcal{O}_{X_{0}}}(\widetilde{X}).$$
Due to the presence of a graded structure on $\widetilde{X}$, $\mathcal{O}_{\widetilde{X}}$ is naturally given a graded algebra structure. Moreover, the pairing between $\E$ and $\E^*[1]$ gives rise to a Poisson bracket structure $\{-,-\}$ of degree 1. Explicitly, the Poisson bracket is determined on generators as
$$
\{ \phi_i,\phi_j\} = 0, \qquad \{ \phi_i,\phi_j^*\} = \delta_{ij} ,\qquad \{ \phi_i^*,\phi_j^*\} = 0.
$$
As a final ingredient for the BV formalism, we come to the extension of the action functional $S_0$ to $\widetilde X$.

\begin{definition}
\label{defn:bv}
Let $(X_{0}, S_{0})$ be a gauge theory. Then an {\em extended theory} associated to $(X_{0}, S_{0})$ is a pair $(\widetilde{X}, \widetilde{S})$, where $\widetilde X$ is a super $\mathbb{Z}$-graded vector space as in \eqref{eqn:Xconfsp}, for $\mathcal{E}$ a $\mathbb{Z}_{\geqslant 0}$-graded locally free $\mathcal{O}_{X_{0}}$-module with homogeneous components of finite rank such that $[\mathcal{E}]^{0} = X_{0}$, and $\widetilde S$ is a $0$-degree element in $\mathcal{O}_{\widetilde{X}}$ such that $\widetilde{S}|_{X_{0}} = S_{0}$, with $\widetilde{S} \neq S_{0}$,
and that solves the `classical master equation' $\{\widetilde{S}, \widetilde{S}\}=0$.
% (i.e. the `classical master equation').

%with $\widetilde X$ as in \eqref{eqn:Xconfsp}above and the extended action $\widetilde S$ as a $0$-degree element in $\mathcal{O}_{\widetilde{X}}$ such that 
%$\widetilde{S}|_{X_{0}} = S_{0}$, with $\widetilde{S} \neq S_{0}$,
%and $\{\widetilde{S}, \widetilde{S}\}=0$ (i.e. the `classical master equation').
%%  where: 
%% \begin{enumerate}[$\blacktriangleright$]
%%  \item The {\em extended configuration space} $\widetilde{X}$ is a super $\mathbb{Z}$-graded vector space, whose $0$-degree homogeneous component is $[\widetilde{X}]^{0}=X_{0}$ and which is suitable of the following decomposition:
%%  \begin{equation}
%%  \label{Eq. decomposizione conf esteso}
%%  \widetilde{X} = \mathcal{E} \oplus \mathcal{E}^{*}[1],
%%  \end{equation}
%% for $\mathcal{E}$ a $\mathbb{Z}_{\geqslant 0}$-graded locally free $\mathcal{O}_{X_{0}}$-module with homogeneous components of finite rank. Moreover, $\mathcal{E}$ describes the fields/ghost fields content of $\widetilde{X}$ while $\mathcal{E}^{*}[1]$ represents the antifields/antighost fields part.
%% \item The {\em extended action} $\widetilde{S}$ is a $0$-degree element in $\mathcal{O}_{\widetilde{X}}$, that is, $\widetilde{S} \in [\mathcal{O}_{\widetilde{X}}]^{0}$, with $\mathcal{O}_{\widetilde{X}}$ the algebra of regular functions on $\widetilde{X}$. Moreover, $\widetilde{S}$ is required to satisfy:
%\begin{itemize}
%\item $\widetilde{S}|_{X_{0}} = S_{0}$, with $\widetilde{S} \neq S_{0}$;
%\item $\{\widetilde{S}, \widetilde{S}\}=0$ (the `classical master equation').
%\end{itemize} 
\end{definition}
We refer to the difference $S_{\BV} =\widetilde{S} - S_{0}$ as the {\em BV action} of the extended theory $(\widetilde{X}, \widetilde{S})$. Note also that, even though each homogeneous component of the graded vector space $\widetilde{X}$ is taken to be finite-dimensional, there is no hypothesis on the number of non-trivial homogeneous components in $\widetilde{X}$ which may be infinite. 

\begin{definition}
Given an extended theory $(\widetilde{X}, \widetilde{S})$, the induced {\em classical BRST cohomology complex} is $(\mathcal{C}^{\bullet}(\widetilde{X}), d_{\widetilde{S}})$, where 
$$
\mathcal{C}^{i}(\widetilde{X}) = [\Sym_{\mathcal{O}_{X_{0}}}(\widetilde{X})]^{i} 
 \qquad  (i \in \Z)
$$
and $d_{\widetilde{S}}:= \{\widetilde{S}, - \}$ is  the coboundary operator.
\end{definition}
The fact that the map $d_{\widetilde{S}}$ defines a linear and graded-derivative operator of degree $1$ over $\mathcal{O}_{\widetilde{X}}$ is a consequence of the properties of the Poisson bracket, whereas $(d_{\widetilde{S}})^{2} =0$ follows from the (graded) Jacobi identity and the fact that $\widetilde{S}$ solves the classical master equation. 

We now describe the gauge-fixing of our gauge theory in the context of the BV formalism. This essentially comes down to removing the anti-fields in the action $\widetilde{S}$; a key role in this construction is played by the choice of a {\em gauge-fixing fermion}.

\begin{definition}
\label{def gauge fixing fermion}
Let $\widetilde{X}= \mathcal{E} \oplus \mathcal{E}^{*}[1]$ be the above extended configuration space. A {\em gauge-fixing fermion} $\Psi$ is defined to be a Grassmannian function $ \Psi \in [\mathcal{O}_{\mathcal{E}}]^{-1}$. 
\end{definition}

From this, given an extended theory $(\widetilde{X}, \widetilde{S})$ together with a gauge-fixing fermion $\Psi$, the corresponding {\em gauge-fixed theory} is a pair $(\widetilde{X}_{\Psi}, \widetilde{S}_{\Psi})$
such that 
$\widetilde{X}_{\Psi} = \mathcal{E}$, where $\mathcal{E}$ is the subspace generated by fields and ghost fields, and 
$$\widetilde{S}_{\Psi} = \widetilde{S}(\varphi_{i}, \varphi_{i}^{*} =\tfrac{\partial \Psi}{\partial \varphi_{i}})$$ so  that $\widetilde{S}_{\Psi} \in [\mathcal{O}_{\mathcal{E}}]^{0}$. 
% with:
%\begin{itemize}
% \item $\widetilde{X}_{\Psi} = \mathcal{E}$, where $\mathcal{E}$ is the subspace generated by fields and ghost fields;
% \item $\widetilde{S}_{\Psi} = \widetilde{S}(\varphi_{i}, \varphi_{i}^{*} =\frac{\partial \Psi}{\partial \varphi_{i}})$ so that $\widetilde{S}_{\Psi} \in [\mathcal{O}_{\mathcal{E}}]^{0}$. 
% \end{itemize}
% \noindent
Given an extended theory $(\widetilde{X}, \widetilde{S})$, the gauge-fixing procedure a priori is not directly applicable, because all the fields/ghost fields in $\widetilde{X}$ have non-negative ghost degree, which impedes the definition of a gauge-fixing fermion for the theory. A solution to this problem was first discovered by Batalin and Vilkovisky \cite{BV1,BV2} who suggested the introduction of {\em auxiliary fields} of negative ghost degree. This is done using so-called {\em trivial pairs}, consisting of fields $B, h$ whose ghost degrees satisfy 
$$
\begin{array}{lr}
\text{deg}(h)= \text{deg}(B) +1 .
% \quad \quad %& \quad \quad \epsilon(h)= \epsilon(B) + 1 \quad  (\mbox{mod } \mathbb{Z}/2). 
\end{array}
$$
Given a trivial pair $(B, h)$, the ghost degrees of the corresponding anti-fields $(B^{*}, h^{*})$ are then related by
$$\begin{array}{lr}
\text{deg}(h^{*})= \text{deg}(B^{*}) - 1.
\end{array}
$$

\begin{definition}
For an extended theory $(\widetilde{X}, \widetilde{S})$ and a trivial pair $(B, h)$, the corresponding {\em total theory} is a pair $(X_{\tot}, S_{\tot})$, where the total configuration space $X_{\tot}$ is the $\mathbb{Z}$-graded vector space generated by $\widetilde{X}$, $B$, $h$ together with the corresponding anti-fields $B^{*}$, $h^{*}$, and the total action is $S_{\tot}= \widetilde{S} + S_{\aux} $ with  $\ S_{\aux}= h B^{*}.$
%\begin{itemize}
% \item $X_{\tot}$ is the $\mathbb{Z}$-graded vector space generated by $\widetilde{X}$, $B, h$ and their corresponding anti-fields $B^{*}, h^{*}$;
%\item $S_{\tot}= \widetilde{S} + S_{\aux},\ $ with $\ S_{\aux}= h B^{*}.$
%\end{itemize}
\end{definition}

In other words, the functional $S_\tot$ is in  the algebra of functionals  $\mathcal{O}_{X_{\tot}}$ on $X_\tot$ that is obtained along the same lines as $\mathcal{O}_{\widetilde{X}}$. Moreover, this algebra carries a graded Poisson structure, determined by the bracket on $\mathcal{O}_{\widetilde{X}}$ and
$$\{ B, B^{*}\} = \{ h, h^{*}\} =1,$$
with all other combinations of the $B,h,B^*,h^*$ among themselves and with other fields being zero. The fact that $S_\tot$ does not depend on $h^*$ or $B$ implies that also $S_\tot$ satisfies the classical master equation and, furthermore, that
$$
\{ \widetilde  S, - \} = \{ S_\tot,-\} 
$$
when we consider $\mathcal{O}_{\widetilde{X}}$ as a subalgebra of $\mathcal{O}_{X_{\tot}}$. In fact, it follows that the classical BRST cohomology complex $(\mathcal{C}^{\bullet}(\widetilde{X}), d_{\widetilde S})$ is quasi-isomorphic to the complex $(\mathcal{C}^{\bullet}(X_{\tot}), d_{S_{\tot}})$, where
$$
\mathcal{C}^{\bullet}(X_{\tot}) = [\Sym_{\mathcal{O}_{X_{0}}}(X_{\tot})]^{i}  \quad (i \in \mathbb{Z})
$$
and $d_{S_{\tot}}: = \big\{ S_{\tot}, - \big\}$. This is the reason for the terminology {\it trivial pairs}. Batalin and Vilkovisky showed in  \cite{BV1,BV2}  that the number of trivial pairs that need to be introduced is determined by the aforementioned level of reducibility $L$ of the gauge theory. 

\begin{theorem}
\label{teorema auxiliary fields caso generale}
Let $(\widetilde{X}, \widetilde{S})$ be an extended theory for a gauge theory with level of reducibility $L$. Then $\widetilde{X}$ is enlarged to give $X_\tot$ by introducing a collection of trivial pairs   $\{ (B_{i}^{j}, h_{i}^{j}) \}$  for $i= 0, \dots, L$ and $j=1,\dots, i+1$  such that $\deg(B_{i}^{j}) = j -i -2 $ if $j$ is odd, or $\deg(B_{i}^{j}) = i -j +1$ if $j$ is even.
%\begin{itemize}
%\item if $j$ is odd: $deg(B_{i}^{j}) = j -i -2 $;
%\item if $j$ is even: $deg(B_{i}^{j}) = i -j +1.$
%\end{itemize}
\end{theorem}
After implementing the gauge-fixing, the pair $(X_{\tot}, S_{\tot})|_{\Psi}$ may still induces a cohomology complex, which is called {\em gauge-fixed BRST cohomology complex}. In fact, while this always happens if the theory is considered on shell, the existence of this cohomology complex in the off-shell case depends on the explicit form of the action $\widetilde{S}$. For completeness, we give its definition.

\begin{definition}
\label{definition of gauge-fixed BRST cohomology}
For a gauge-fixed theory  {\small{$(X_{\tot}, S_{\tot})|_{\Psi}$}} with {\small{$X_\tot = Y \oplus Y^*[1]$}}, the corresponding {\em gauge-fixed BRST cohomology complex} is $(\mathcal{C}^{\bullet}(Y), d_{S_{\tot}}|_{\Psi})$, where 
$$
\mathcal{C}^{k}(Y) = [\Sym_{\mathcal{O}_{X_{0}}}(Y)]^{k} \qquad (k \in \mathbb{Z})
$$ 
and the coboundary operator is given by $d_{S_{\tot}}|_{\Psi} := \big\{ S_{\tot}, - \big\}|_{\Sigma_{\Psi}}$ for the %Lagrangian
 submanifold  $\Sigma_{\Psi}$  of $X_{\tot}$ defined by the gauge-fixing conditions $\varphi_{i}^{*} = \frac{\partial \Psi}{\partial \varphi_{i}}.$
\end{definition}

\section[The BV construction applied to a U(2)-model]{The BV construction applied to a $U(2)$-model}
\label{algorithm applied to matrix model}
\noindent
The BV construction, reviewed in the previous section, will now be applied to a gauge theory naturally induced by a finite spectral triple on the algebra $M_{n}(\mathbb{C})$. Indeed, by the construction of Proposition \ref{Prop: gauge theory from spectral triple}, we have that the finite spectral triple 
$$ (M_{n}(\mathbb{C}), \mathbb{C}^{n}, D), $$
for an hermitian $n \times n$-matrix $D$,  yields a gauge theory $(X_{0}, S_{0})$ with gauge group $\mathcal G$ such that 
$$X_0=\{ M \in M_n(\mathbb C): M^{*} = M \},
\quad  S_0[M] = \tr f (M)
 \quad  \text{and} \quad  \mathcal G = U(n),$$
with $f$ a polynomial in $\Pol_{\mathbb{R}}(x) $ and the adjoint action of $\mathcal G$ on $X_0$. For simplicity, we will analyze the result of applying the BV construction on this model for $n=2$. To proceed with the construction, first fix a basis for $X_{0}$ given by Pauli matrices (together with the identity matrix):
\begin{equation}
\sigma_{1}= \begin{pmatrix}
		    0 & 1\\
		    1 & 0
	\end{pmatrix}
, \quad 
\sigma_{2}= \begin{pmatrix}
		    0 & -i\\
		    i & 0
		\end{pmatrix}
, \quad 
\sigma_{3}= \begin{pmatrix}
		    1 & 0\\
		    0 & -1
		\end{pmatrix}	 
, \quad
\sigma_{4}= \begin{pmatrix}
		    1 & 0\\
		    0 & 1
		\end{pmatrix}.
\label{basis}
\end{equation}
%Thus $X_{0}\cong \mathbb{A}^{4}_{\mathbb{R}},$ with $\mathbb{A}^{4}_{\mathbb{R}}$ the $4$-dimensional real affine space. 
Denoting by $\left\lbrace M_{a}\right\rbrace_{a=1}^4$ the dual basis of $\left\lbrace\sigma_{a}\right\rbrace_{a=1}^4$, $X_{0}$ is isomorphic to a $4$-dimensional real vector space generated by four independent initial fields:
$$X_{0}\simeq \langle M_{1}, M_{2}, M_{3}, M_{4}\rangle_{\mathbb{R}}.$$
\noindent 
Hence the ring of regular functions on $X_{0}$ is the ring of polynomials in the variables $M_{a}$, $\mathcal{O}_{X_{0}}= \Pol_{\mathbb{R}}(M_{a})$. In terms of the coordinates $M_{a}$ the spectral action $S_{0}$, defined by a polynomial $f=\sum_{i=0}^{r}\mu_{i}x^{i}$, takes the following explicit form:
$$\begin{array}{ll}
 S_{0}\hspace{-2mm}& = 2 \left[ \sum_{a=0}^{\lfloor r/2\rfloor} \mu_{2a}\left( \sum_{s=0}^{a}\binom{2a}{2s} (M_{1}^2 + M_{2}^2 + M_{3}^ 2)^{a-s} M_{4}^{2s}\right)\right. \\
[1.5ex] & \quad+ \left. \sum_{a=0}^{\lceil r/2\rceil -1} \mu_{2a+1}\left( \sum_{s=0}^{a}\binom{2a+1}{2s+1} (M_{1}^2 + M_{2}^2 + M_{3}^ 2)^{a-s} M_{4}^{2s+1}\right)
\right] .
\end{array}
$$
However, an action $S_{0}$ of this type only represents a family of $U(2)$-invariant functionals on $X_{0}$. In fact, the most general form for a functional $S_{0}$ on $X_{0}$ that is invariant under the adjoint action of the gauge group $U(2)$ is as symmetric polynomial in the eigenvalues $\lambda_1, \lambda_2$ of the variable $M \in X_0$, or, equivalently, as polynomial in the symmetric elementary polynomials $a_1 = \lambda_1+\lambda_2$ and $a_2=\lambda_1 \lambda_2$. In terms of the coordinates $M_a$ we have:
$$
\lambda_i= M_{4}\pm \sqrt{M_{1}^2+M_{2}^2+M_{3}^2}, \quad
a_{1}= 2M_4,  \quad
a_{2}=  M^{2}_{4}- (M_{1}^2+M_{2}^2+M_{3}^2).
$$
Hence the generic form for a $U(2)$-invariant action $S_0 \in \Pol_{\mathbb{R}}(M_{a})$ is
\begin{equation}
S_{0} = \sum_{k=0}^{r} \mbox{ }(M_{1}^2 + M_{2}^2 + M_{3}^2)^k g_{k}(M_4),
\label{S_0 generale}
\end{equation}
where $g_k(M_{4})\in \Pol_{\mathbb{R}}(M_4)$. Because the introduction of extra (non-physical) fields is motivated by the necessity of eliminating the symmetries in the action functional $S_0$, the BV construction may give rise to different extended configuration spaces $\widetilde{X}$, depending on the explicit form of $S_{0}$. For our $U(2)$-matrix model we have three different cases: 
\begin{enumerate}
\item If {\small{$S_{0} \in \Pol_{\mathbb{R}}(M_{4})$}}, there are no symmetries that need to be removed by adding ghost fields. Hence, the construction of $\widetilde{X}$ stops at the first stage, after including the anti-fields corresponding to the initial fields in $X_{0}$:
{\small{$$\widetilde{X} = X_{0} \oplus \langle M^{*}_{1}, M^{*}_{2}, M^{*}_{3}, M^{*}_{4}\rangle_{-1}.$$}}
\item If {\small{$GCD(\partial_{1}S_{0},  \partial_{2}S_{0}, \partial_{3}S_{0}, \partial_{4}S_{0}) = 1$}}, three independent ghost fields $C_{1}$, $C_{2}$, $C_{3}$ are inserted to compensate for the three independent relations existing over $\mathcal{O}_{X_{0}}$ between pairs of partial derivatives of $S_{0}$:\vspace{2mm}\\
{\small{$M_{1}({\partial_{2}}S_{0}) =M_{2}({\partial_{1}}S_{0}),  \ \
M_{1}({\partial_{3}}S_{0})=M_{3}({\partial_{1}}S_{0}),  \ \   M_{2}({\partial_{3}}S_{0})=M_{3}({\partial_{2}}S_{0}).\vspace{2mm}$}}
\\
After having eliminated these three symmetries, there is still one relation that involves all three terms $\partial_{1}S_{0}$, $\partial_{2}S_{0}$, and $\partial_{3}S_{0}$. Hence, we have to add a ghost field $E$ of ghost degree $2$.
\item If {\small{$GCD(\partial_{1}S_{0}, \partial_{2}S_{0}, \partial_{3}S_{0}, \partial_{4}S_{0}) = D \notin \mathbb{R}$}}, the action $S_{0}$ presents additional symmetries to cancel and so the extended configuration space $\widetilde{X}$ has to be further enlarged, obtaining that
\vspace{1mm}\\
{\small{$\begin{array}{ll}
\widetilde{X} = & \langle K^* \rangle_{-4} \oplus \langle E^{*}_{1}, \dots, E^*_4\rangle_{-3} \oplus \langle C^{*}_{1}, \cdots, C^{*}_{6}\rangle_{-2} \oplus \langle M^*_{1}, \dots, M^*_{4}\rangle_{-1}\\
[.8ex]
&  \oplus X_{0} \oplus \langle C_{1}, \cdots, C_{6} \rangle_{1} \oplus \langle E_{1}, \dots, E_4\rangle_{2} \oplus \langle K\rangle_{3}.
\end{array}$}}
\end{enumerate} 
\noindent
Here we focus on the generic situation (2), for which we have the following result ({\it cf.} \cite{primo_articolo}). 
\begin{theorem}
\label{Theorem: BV variety model}
Let $(X_{0}, S_{0})$ be a gauge theory with, as configuration space, $X_{0}\simeq \langle M_a \rangle_{\mathbb R}$ for $a=1, \dots, 4$, and, as action functional, $S_{0}\in \mathcal{O}_{X_{0}}$ of the form \eqref{S_0 generale}. If {\small{$GCD(\partial_{1}S_{0}, \partial_{2}S_{0}, \partial_{3}S_{0}, \partial_{4}S_{0}) = 1$}}, then
the minimally extended configuration space $\widetilde{X}$ is the following $\mathbb{Z}$-supergraded real vector space:
$$\widetilde{X} = \langle E^{*}\rangle_{-3} \oplus \langle C^*_{1}, C^*_{2}, C^*_{3}\rangle_{-2} \oplus \langle M^*_{1}, \dots, M^*_{4}\rangle_{-1} \oplus X_{0} \oplus \langle C_{1}, C_{2}, C_{3} \rangle_{1} \oplus \langle E\rangle_{2}.$$
Moreover, the general solution of the classical master equation on $\widetilde{X}$ that is linear in the anti-fields, of at most degree $2$ in the ghost fields and with coefficients in $\mathcal{O}_{X_{0}}$ is given by $\widetilde{S} =  S_{0}+S_\BV$, for \vspace{-3mm}

%{\small{
%\begin{multline}
%\label{eq: generic extended action}
%S_\BV= M^{*}_{1}(\omega M_{2}C_{3}-\mu M_3C_{2}) + M^{*}_{2}(\lambda M_3C_{1}- \omega M_{1}C_{3}) +  M^{*}_{3}(\mu M_{1}C_{2} -\lambda M_2C_{1} ) \\
%+ C^*_{1}\left(\alpha \mu \omega M_1E + \mu M_{1} M_{3}T C_{1}C_{2}- \omega M_{1} M_{2}TC_{1}C_{3} + \tfrac{\omega \mu }{\lambda} (1 +  M_{1}^{2}T)C_{2}C_{3} \right) \\
%+ C^*_{2}\left(\alpha \lambda \omega M_2E + \lambda M_{2} M_{3}T C_{1}C_{2}   + \tfrac{\lambda \omega}{\mu}(-1- M_{2}^2 T)C_{1}C_{3}  + {\omega} M_{1}M_{2}T  C_{2}C_{3}\right)\\
%+ C^*_{3}\left(\alpha \lambda \mu M_3E + \tfrac{\mu \lambda}{\omega}(1 + M_{3}^2 T)C_{1}C_{2} - \lambda M_{2}M_{3} T C_{1}C_{3}   + \mu M_{1} M_{3} T C_{2}C_{3}\right)
%\end{multline}}}
%S_\BV= M^{*}_{1}(\omega M_{2}C_{3}-\mu M_3C_{2}) + M^{*}_{2}(\lambda M_3C_{1}- \omega M_{1}C_{3}) +  M^{*}_{3}(\mu M_{1}C_{2} -\lambda M_2C_{1} ) \\
%[.5ex]
%\quad \quad \ + C^*_{1}\big(\alpha \mu \omega M_1E + \mu M_{1} M_{3}T C_{1}C_{2}- \omega M_{1} M_{2}TC_{1}C_{3} + \tfrac{\omega \mu }{\lambda} (1 +  M_{1}^{2}T)C_{2}C_{3} \big) \\
%[.5ex]
%\quad  \quad\ + C^*_{2}\big(\alpha \lambda \omega M_2E + \lambda M_{2} M_{3}T C_{1}C_{2}   + \tfrac{\lambda \omega}{\mu}(-1- M_{2}^2 T)C_{1}C_{3}  + {\omega} M_{1}M_{2}T  C_{2}C_{3}\big)\\
%[.5ex]
%\quad \quad \ + C^*_{3}\big(\alpha \lambda \mu M_3E + \tfrac{\mu \lambda}{\omega}(1 + M_{3}^2 T)C_{1}C_{2} - \lambda M_{2}M_{3} T C_{1}C_{3}   + \mu M_{1} M_{3} T C_{2}C_{3}\big)
%\begin{equation}
%\begin{array}{l}
\begin{multline}
S_{\BV}= \sum_{i, j, k} \epsilon_{ijk}\alpha_{k}M_{i}^{*}M_{j}C_{k} + \sum_{i, j,k} C_{i}^{*}\big[\tfrac{\alpha_{j}\alpha_{k}}{2\alpha_{i}} (\beta\alpha_{i}M_{i}E + \epsilon_{ijk}C_{j}C_{k}) \\
+ M_{i}T \big( \sum_{a, b,c} \epsilon_{abc}\tfrac{\alpha_{b}\alpha_{c}}{2\alpha_{i}}M_{a}C_{b}C_{c}\big)\big]
\label{eq: generic extended action}
\end{multline}
where $\alpha_{i}, \beta\in \mathbb{R}\backslash \left\lbrace 0 \right\rbrace$, $T \in \Pol_{\mathbb{R}}(M_a)$, and $\epsilon_{ijk}$ ($\epsilon_{abc}$) is the totally anti-symmetric tensor in three indices $i, j, k \in \{ 1, 2, 3 \}$ ($a, b, c \in \{ 1, 2,3\}$) with $\epsilon_{123}=1$.
\end{theorem}

Once the extended theory {\small{$(\widetilde{X}, \widetilde{S})$}} has been constructed, another step is needed to be able to implement the gauge-fixing procedure, namely, we have to introduce the auxiliary fields. Because {\small{$\widetilde{X}$}} contains ghost fields of at most ghost degree $2$, the pair {\small{$(\widetilde{X}, \widetilde{S})$}} describes a reducible theory with level of reducibility {\small{$L=1$}}. Hence, according to Theorem \ref{teorema auxiliary fields caso generale}, the extended configuration space {\small{$\widetilde{X}$}} has to be enlarged by adding three trivial pairs 
$$\{(B_{i}, h_{i})\}_{i= 1, 2, 3}\quad  \mbox{with} \quad \deg(B_{i})=  -1,  \quad \mbox{and} \quad
\deg(h_{i})= 0,$$
which correspond to the three ghost fields {\small{$C_{i}$}}, together with the two trivial pairs {\small{$(A_{1}, k_{1})$}} and {\small{$(A_{2}, k_{2})$}}, corresponding to the ghost field $E$ and satisfying
$$ \deg(A_{1})= -2,\quad  \deg(k_{1})= -1, \quad \deg(A_{2})= 0, \quad \deg(k_{2})= 1. $$
The total theory {\small{$(X_{\tot}, S_{\tot})$}} now also includes the above auxiliary fields and is given by a $\mathbb{Z}$-graded vector space $X_{tot} = Y \oplus Y^{*}[1]$, with
%\begin{align*}
%X_{\tot}  &= \langle E^{*} \rangle \oplus \langle C^{*}_{i}, A_{1}, k_{2}^{*}\rangle \oplus \langle M^{*}_{a}, B_{i}, h^*_{i}, k_{1}, A_{2}^* \rangle \\
%&\quad \oplus \langle M_{a}, B_{i}^*,  h_{i}, k_{1}^*, A_{2}\rangle  \oplus \langle C_{i}, A_{1}^*, k_{2}\rangle \oplus \langle E\rangle, 
%\end{align*}
$$ Y= \langle M_{a}, B_{i}^*,  h_{i}, k_{1}^*, A_{2}\rangle_{0}  \oplus \langle C_{i}, A_{1}^*, k_{2}\rangle_{1} \oplus \langle E\rangle_{2}, $$
for $a=1, \dots, 4$, $i=1, 2, 3$, together with an $S_{\tot}= \widetilde{S} +S_{\aux} $, where
\begin{equation}
\label{eq: azione triv}
S_{\aux} := \sum_{i=1}^{3}B^{*}_{i} h_{i} + \sum_{j=1}^{2} A^{*}_{j}k_{j}.
\end{equation}

\section[NCG and the BV approach]{The BV approach in the framework of NCG}
\label{chapter: NCG and the BV approach}
\subsection{The BV spectral triple}
We now formulate the BV construction for the above {\small{$U(2)$}}-matrix model in terms of noncommutative geometry. That is, we describe the extended theory {\small{$(\widetilde{X}, \widetilde{S})$}} by means of a spectral triple for which the fermionic action yields {\small{$\widetilde{S}$}}. In order to simplify the computation, we consider the pair {\small{$(\widetilde{X}, \widetilde{S})$}} as described in Theorem \ref{Theorem: BV variety model}, but where in formula (\ref{eq: generic extended action}) we take the polynomial {\small{$T=0$}} and set the real coefficients {\small{$\alpha_{i}=\beta=1$}}. Hence, we analyze the case when the action $S_\BV$ has the following form:
{\small{
\begin{gather}
 S_\BV:=M^{*}_{1}(-M_3C_{2}+ M_{2}C_{3})+M^{*}_{2}(M_3C_{1} - M_{1}C_{3}) +  M^{*}_{3}(- M_2C_{1}+ M_{1}C_{2})  \nonumber \\
 + C^*_{1}(M_1E  + C_{2}C_{3}) + C^*_{2}(M_2E - C_{1}C_{3}) + C^*_{3}( M_3E + C_{1}C_{2}).
\label{eq: S-BV}
\end{gather}
}}
The construction of the so-called {\em BV spectral triple}
 $$(\mathcal{A}_{\BV}, \mathcal{H}_{\BV}, D_{\BV}, J_{\BV})$$ 
 proceeds in steps, where the form of the algebra $\A_{\BV}$ is determined as the last ingredient. 

\subsection*{The Hilbert space $\H_\BV$}
We let $\mathcal{H}_{\BV}$ be the following Hilbert space:
 $$\mathcal{H}_{\BV} = \mathcal{H}_{M} \oplus \mathcal{H}_{C} := M_{2}(\mathbb{C}) \oplus M_{2}(\mathbb{C}), $$
where the subscripts $M$ and $C$ refer to the gauge fields and ghost fields; this will be justified below. The inner product structure is given as usual by the the Hilbert--Schmidt inner product on each summand $M_{2}(\mathbb{C})$, that is to say, $
\langle - , - \rangle:  \mathcal{H}_{\BV} \times \mathcal{H}_{\BV} \rightarrow  \mathbb{C}$, with  
$$ \langle (\phi_M, \phi_C) , (\phi'_M, \phi'_C)\rangle = \tr(\phi_M (\phi'_M)^*) + \tr(\phi_C (\phi'_C)^*), $$
%$$\begin{array}{lccl}
%\langle \quad ,\quad\rangle: &  \mathcal{H}_{\BV} \times \mathcal{H}_{\BV} & \longrightarrow & \mathbb{C} \\
%[.5ex]
%& \left( (\phi_M, \phi_C) , (\phi_M, \phi_C) \right) & \mapsto & \tr(\phi_M (\phi'_M)^*) + \tr(\phi_C (\phi'_C)^*),
%  \end{array}
%$$
for $\phi_M$, $\phi_M' \in\mathcal{H}_{M}$, $\phi_C$, $\phi_C' \in\mathcal{H}_{C}$. Taking the orthonormal basis of $M_2(\C)$ given in (\ref{basis}) we can of course identify 
 $$\mathcal{H}_{\BV} \cong \langle m_{1}, m_{2}, m_{3}, e\rangle \oplus \langle c_{1}, c_{2}, c_{3}, c_{4}\rangle \cong \mathbb{C}^{8}, $$
in terms of which the inner product reads
$$\langle \varphi , \psi \rangle = \sum_{a =1}^{3} \overline{m}_{a, \varphi} m_{a, \psi} + \bar{e}_{\varphi} e_{\psi}+ \sum_{j = 1}^{4} \bar{c}_{j, \varphi} c_{j, \psi}.
$$

\begin{remark}
The Hilbert space $\mathcal{H}_{\BV}$ has also another possible decomposition as direct sum of two vector spaces: $\mathcal{H}_{\BV} = \mathcal{H}_{BV, f} \oplus i \cdot \mathcal{H}_{BV, f},$ 
 with 
 \begin{equation}
 \begin{array}{ll}
 \label{decomposition Hilbert space H_BV}
 \mathcal{H}_{BV, f} & = [i \cdot su(2) \oplus u(1)] \oplus i \cdot u(2)\\
 [1ex] & \simeq \langle M_{1}, M_{2}, M_{3}, iE \rangle_{\mathbb{R}} \oplus \langle C_{1}, C_{2}, C_{3}, C_{4}\rangle_{\mathbb{R}}\ .
 \end{array}
  \end{equation}
In (\ref{decomposition Hilbert space H_BV}) we denote the real part of the complex variables $m_{a}$ and $c_{j}$ by $M_{a}$ and $C_{j}$, respectively, while $E$ is the imaginary part of the complex variable $e$. This choice of notation is motivated by the fact that these variables coincide with the gauge fields and ghost fields that generate the positively graded part of the extended configuration space $\widetilde{X}$, as we will see in Theorem \ref{teorema per BV spectral triple} below. The fourth ghost field $C_{4}$ will not enter the fermionic action $S_{\ferm}$ as it decouples, being consistent with $\widetilde{X}$ having only three ghost fields $C_{i}$ in our model.
\end{remark}

\subsection*{The real structure $J_{\BV}$}
Up to this point ---with the algebra $\A_\BV$ and self-adjoint operator $D_\BV$ yet to be determined--- a real structure is simply given by an anti-linear isometry $J_{\BV}:  \mathcal{H}_{\BV} \rightarrow \mathcal{H}_{\BV}$, which we take to be
%$$\begin{array}{cccl}
% J_{\BV}:  & \mathcal{H}_{\BV} & \longrightarrow &  \mathcal{H}_{\BV}\\
% [.5ex]
% & (\phi_{M}, \phi_{C}) & \mapsto & J_{\BV}(\phi_{M}, \varphi_{C}): = i \cdot (\phi_M^{*}, \phi^{*}_C) \ .
%\end{array}
%$$
$$J_{\BV}(\phi_{M}, \varphi_{C}): = i \cdot (\phi_M^{*}, \phi^{*}_C) $$
for $\phi_M \in \H_M, \varphi_C \in \H_C$. In terms of the basis \eqref{basis} we have for $\phi\in \H_\BV$
$$
J_{\BV}(\varphi):= i \cdot [\bar{m}_{1}, \bar{m}_{2}, \bar{m}_{3}, \bar{e}, \bar{c}_{1}, \bar{c}_{2}, \bar{c}_{3}, \bar{c}_{4}]^{T}.$$

\subsection*{The linear operator $D_{\BV}$}
The self-adjoint linear operator $D_{\BV}$ acting on the Hilbert space $\mathcal{H}_{\BV}$ is given by the following expression 
$$
D_\BV := \begin{pmatrix} T & R \\ R^*& S \end{pmatrix}
$$
in terms of the decomposition $\H_\BV= \H_M \oplus \H_C$. The linear operators $R,S,T$ are defined by
\begin{equation*}
\begin{aligned}
R: \H_C &\to \H_M ;\\
\phi_C &\mapsto [\beta, \phi_C ],
\end{aligned}\quad
\begin{aligned}
S: \H_C &\to \H_C ;\\
\phi_C &\mapsto [\alpha, \phi_C],
\end{aligned}
\quad\begin{aligned}
T: \H_M &\to \H_M ;\\
\phi_C &\mapsto [\alpha, \phi_C ]_+,
\end{aligned}
\end{equation*}
where $\alpha$ and $\beta$ are hermitian, traceless $2 \times 2$-matrices. We stress that thus $R$ and $S$ are derivations of $M_2(\C)$, but that $T$ is an odd derivation given in terms the anti-commutator. 

We can write $\alpha$ and $\beta$ in terms of the Pauli matrices as follows
\begin{align*}
\alpha &= \tfrac{1}{2}\big[(-C_{1}^{*})\sigma_1 + (-C_{2}^{*})\sigma_2 + (-C_{3}^{*})\sigma_3\big]\\
\beta &= \tfrac{1}{2}\big[(-M_{1}^{*})\sigma_1 + (-M_{2}^{*})\sigma_2 + (-M_{3}^{*})\sigma_3\big],
\end{align*}
where $C_i^*$ and $M_i^*$ are real variables. Then, in terms of the orthonormal basis \eqref{basis} for $\H_M$ and $\H_C$, we find the following $4 \times 4$-matrices for $R,S,T$: 
%\begin{gather*}
%R:=  \begin{pmatrix}0 & +i M_3^* & -i M_2^* & 0 \\
%-iM_3^* & 0 & +iM_1^* & 0 \\
%+iM_2^* &-i M_1^* & 0 & 0 \\
%0 & 0 & 0 &0\end{pmatrix}, \qquad 
%S:=  \begin{pmatrix}0 & +i C_3^* & -i C_2^* & 0 \\
%-iC_3^* & 0 & +iC_1^* & 0 \\
%+iC_2^*& -i C_1^* & 0 & 0 \\
%0 & 0 & 0 &0\end{pmatrix}\\
%T:= \begin{pmatrix} 0 & 0 & 0 & C_1^* \\
% 0 & 0 & 0 & C_2^*\\
%0 & 0 & 0 & C_3^* \\
%C_1^* & C_2^* & C_3^* & 0
%\end{pmatrix}
%\end{gather*}
{\small{
\begin{gather*}
R:=  \begin{pmatrix}0 & +i M_3^* & -i M_2^* & 0 \\
-iM_3^* & 0 & +iM_1^* & 0 \\
+iM_2^* &-i M_1^* & 0 & 0 \\
0 & 0 & 0 &0\end{pmatrix}, \qquad 
S:=  \begin{pmatrix}0 & +i C_3^* & -i C_2^* & 0 \\
-iC_3^* & 0 & +iC_1^* & 0 \\
+iC_2^*& -i C_1^* & 0 & 0 \\
0 & 0 & 0 &0\end{pmatrix}
\end{gather*}
}}
{\small{
\begin{gather*}
T:= \begin{pmatrix} 0 & 0 & 0 & C_1^* \\
 0 & 0 & 0 & C_2^*\\
0 & 0 & 0 & C_3^* \\
C_1^* & C_2^* & C_3^* & 0
\end{pmatrix}
\end{gather*}
}}

Of course, the notation used for the components of $\alpha$ and $\beta$ has been chosen with purpose: indeed, we will prove that upon inserting all anti-fields in the linear operator $D_{\BV}$, the corresponding fermionic action yields the BV action $S_{\BV}$. 

It is not true that the above $D_\BV$ commutes or anti-commutes with $J_\BV$. Instead, we may decompose $D_\BV$ as
$$ D_\BV = D_1 +D_2 \quad \text{with}\quad D_1 = \begin{pmatrix} 0 & R \\ R^*& S \end{pmatrix} , \qquad D_2 = \begin{pmatrix} T & 0 \\ 0& 0 \end{pmatrix}.
$$
for which we find that
\begin{align*}
J_{\BV}D_{1} &= - D_{1}J_{\BV},\\
 J_{\BV}D_{2} &= + D_{2}J_{\BV}.
\end{align*}
In anticipation of what is to come, this suggests that a real spectral triple of {\em mixed KO-dimension} will appear.  

\subsection*{The algebra $\A_{\BV}$}
We now come to the final ingredient of the BV spectral triple which is the algebra $\mathcal{A}_{\BV}$. We take it to be largest unital subalgebra of the algebra of all linear operator $\mathcal L(\H_\BV)$ that satisfies the {\em commutation rule} and {\em first-order condition} of Definition \ref{def real structure}.

\begin{lma}
Let $\H_\BV,J_\BV$ and $D_\BV$ be as defined above. Then the maximal unital subalgebra $\tilde \A$ of $\mathcal L(\H_\BV)$ that satisfies
$$
[a,J_\BV b^* J^{-1}_\BV]= 0, \qquad [[D_\BV,a],J_\BV b^* J^{-1}_\BV]= 0; \qquad (a,b\in \tilde \A)
$$
is given by $\tilde \A = M_2(\C)$ acting diagonally on $\H_\BV$. 
\end{lma}
\proof
The commutation rule $[a,J_\BV b^* J^{-1}]= 0$ for all $a,b \in \tilde\A$ implies that $\H_\BV$ carries an $\tilde\A$-bimodule structure. This already restricts $\tilde\A$ to be a subalgebra of $M_2(\C) \oplus M_2(\C)$, acting diagonally on $\H_M \oplus \H_C$. Then, by a staightforward computation of the double commutator {\small{$\left[[D,(a_1,a_2)], J_\BV (b_1,b_2)J_\BV^{-1} \right]$}}, it follows that the first-order condition implies $a_1=a_2$ and $b_1=b_2$. 
%% For the first-order condition, we compute the double commutator
%% $$
%% \left[[D,(a_1,a_2)], J_\BV (b_1,b_2)J_\BV^{-1} \right]= \begin{pmatrix}
%% 0 & R a_2 -a_1 R \\ 
%% R^* a_1 - a_2 R^* & 0 
%% \end{pmatrix}
%% $$
%% which for generic $R$ forces $a_1=a_2$. 
This selects the subalgebra $M_2(\C)$ in $M_2(\C) \oplus M_2(\C)$ as the maximal subalgebra for which both of the above conditions are satisfied. 
\endproof
We will denote this maximal subalgebra by $\A_\BV$. We now make the encountered phenomenon of mixed KO-dimension more precise by the following result. 
\begin{prop}
\label{spectral triple fermionica}
With the above notation,
\begin{itemize}
\item[(i)]
$(\mathcal{A}_{\BV}, \mathcal{H}_{\BV}, D_{1}, J_{\BV})$ is a real spectral triple of KO-dimension $1$.
\item[(ii)] 
 $(\mathcal{A}_{\BV}, \mathcal{H}_{\BV}, D_{2}, J_{\BV})$ is a real spectral triple of KO-dimension $7$.
\end{itemize}
\end{prop}
Before continuing with our BV spectral triple, we develop some new theory on real spectral triple with mixed KO-dimension. 

\begin{definition}
For $(\mathcal{A}, \mathcal{H}, D)$ a finite spectral triple and $J$ an anti-linear isometry on $\mathcal{H}$, we say that $(\mathcal{A}, \mathcal{H}, D, J)$ defines a {\em real spectral triple with mixed KO-dimension} if $J$ satisfies  
$$J^{2}= \pm Id \qquad \mbox{ and } \qquad [a, Jb^{*}J^{-1}] =0$$
for $ a, b \in \mathcal{A}$, the operator $D$ can be seen as a sum $D=D_1 + D_2$ of two self-adjoint operators $D_{1}$, $D_{2}$, which {\em anti}-commutes and commutes, respectively, with $J$:
$$D_1 = - D_1 J  \quad \mbox{ and }  \quad JD_2 = D_2J, $$
and, finally, the first-order condition holds:
$$[[D, a], Jb^{*}J^{-1}] = 0, \qquad ( a, b \in \mathcal{A}).$$
%the following conditions are satisfied:
%\begin{enumerate}[(i)]
%\item $J^{2}= \pm Id$;
%\item there exist two self-adjoint operators $D_{1}$, $D_{2}$ on $\mathcal{H}$ such that:
%$$D=D_1 + D_2, \quad \mbox{ with } \quad \quad JD_1 = - D_1 J  \quad \mbox{ and }  \quad JD_2 = D_2J; $$
%\item $[a, Jb^{*}J^{-1}] =0; \qquad ( a, b \in \mathcal{A})$;
%\item $[[D, a], Jb^{*}J^{-1}] = 0;\qquad ( a, b \in \mathcal{A})$.
%\end{enumerate}
\end{definition}
The notion of {\em mixed KO-dimension} generalizes the usual notion of KO-dimension for real spectral triples allowing the operator $D$ not to fully commute or anti-commute with the isometry $J$. We notice that, if we are considering a genuinely mixed KO-dimension, that is, if both $D_{1}, D_{2} \neq 0$, then the even case is not allowed. 

\begin{prop}
 Let $(\mathcal{A}, \mathcal{H}, D, J)$ be a real spectral triple of mixed KO-dimension. If $J^2 = +1$ then for all $\phi,\psi\in \H$:
\begin{enumerate}[(1)]
 \item the expression $ \mathfrak{A}_{D_{1}}(\varphi, \psi) := \langle J\varphi, D_{1} \psi \rangle$ defines an {\em anti-symmetric} bilinear form on $\mathcal{H}$;
\item the expression $ \mathfrak{A}_{D_{2}}(\varphi, \psi) := \langle J\varphi, D_{2} \psi \rangle$ defines a {\em symmetric} bilinear form on $\mathcal{H}$.
\end{enumerate}
On the contrary, if $J^2 = -1$, then $\mathfrak{A}_{D_{1}}$ is symmetric and $\mathfrak{A}_{D_{2}}$ is anti-symmetric.
\end{prop}

\begin{proof}

(1) Bilinearity of $\mathfrak{A}_{D_{1}}$ is a consequence of $J$ being an anti-linear map, $D_{1}$ being a linear operator and the inner product being anti-linear in its first component and linear in the second. For the anti-symmetry, under the assumption that $J^ 2 = \epsilon$ we compute
$$
\langle J\varphi, D_{1} \psi \rangle 
= \epsilon \langle J\varphi, J^ 2 D_{1} \psi \rangle
=  \epsilon \langle J D_{1} \psi ,\varphi\rangle
=- \epsilon \langle D_{1}  J\psi ,\varphi\rangle
=- \epsilon \langle  J\psi ,D_1 \varphi\rangle
$$
using the anti-commutation of $D_1$ with $J$ and $D_1$ being a self-adjoint operator. \\
\noindent
(2) This follows {\it mutatis mutandis} from (1), assuming $D_2$ to commute with \nolinebreak $J$. \qedhere
\end{proof}

We now return to the BV spectral triple $(\A_\BV,\H_\BV,D_\BV,J_\BV)$ that we constructed for our $U(2)$-matrix model.  

 \begin{theorem}
 \label{teorema per BV spectral triple}
The data $(\mathcal{A}_{\BV}, \mathcal{H}_{\BV}, D_{\BV}, J_{\BV})$ defined above is a real spectral triple with mixed KO-dimension. Moreover, the fermionic action corresponding to the operator $D_{\BV}$ coincides with the BV action in \eqref{eq: S-BV}, {\it i.e.},
$$S_{\BV} = \frac{1}{2} \langle J_{\BV}(\varphi), D_{\BV} \varphi \rangle, \qquad \text{ with } \quad \varphi \in \mathcal{H}_{BV, f}.$$
Here we interpret the variables that parametrize $D$ and the vector $\phi \in \H_\BV$ as follows:
\begin{itemize}
\item $M_a, E$ and $C_j^*$ are real variables;
\item $M_a^*$ and $C_j$ are Grassmannian variables.
\end{itemize}
 \end{theorem}

\begin{proof}
The first claim is an immediate consequence of Proposition \ref{spectral triple fermionica}. The last statement follows by a straightforward computation. 
 \end{proof}

The mixed KO-dimension in the BV spectral triple arises from the particular behavior of the real structure with the operator $D_{\BV}$, which partially commutes and partially {\em anti}-commutes with $J_{\BV}$. However, this is not due to the fact that the BV spectral triple is a direct sum of real spectral triples of different KO-dimensions. Indeed, then the structure of the real spectral triple would be 
$$(\mathcal{A}_{1} \oplus \mathcal{A}_{2}, \mathcal{H}_{1} \oplus \mathcal{H}_{2}, D_{1} \oplus D_{2}, J_{1} \oplus J_{2}).$$
The appearance of a direct sum of the two algebras is in contrast with the structure of $\A_\BV$ as a simple algebra. As a matter of fact, the mixed KO-dimension has a different significance, allowing us to detect the difference in {\em parity} of the components of the BV spectral triple, as we will now explain. 

Namely, the fields that parametrize the Dirac operator and/or represent vectors in Hilbert space are seen to be structured as follows: 
\begin{itemize}
   \item The anti-fields/anti-ghost fields $M^{*}_{a}$ and $C_{j}^{*}$ appear as entries of the operator $D_{\BV}$ while the fields/ghost fields $M_{a}$, $C_{j}$, and $E$ are the components of the vectors in the subspace $\mathcal{H}_{BV, f}$. 
 \item The parities of the fields/ghost fields and anti-fields/anti-ghost fields in the BV spectral triple are a consequence of the structure of the real spectral triple. Indeed, the parities chosen in Theorem \ref{teorema per BV spectral triple} are precisely those for which both $D_{1}$ and $D_{2}$ give a non-trivial contribution to the fermionic action $S_{\ferm}$. 
\end{itemize}

 \subsection{The BV auxiliary spectral triple}
\label{section:The BV auxiliary spectral triple}
In addition to the BV construction, also the technical procedure of introducing auxiliary fields in our $U(2)$-matrix model can be expressed in terms of a spectral triple: in this section, we construct the so-called {\em BV auxiliary spectral triple} 
$$(\mathcal{A}_{\aux}, \mathcal{H}_{\aux}, D_{\aux}, J_{\aux}),$$
for which the fermionic action coincides with the auxiliary action $S_{\aux}$. We follow the same strategy as for the BV spectral triple: the anti-fields {\small{$\{B^{*}_{j}\}$}}, for $j=1, 2, 3$, and {\small{$\{A^{*}_{l}\}$}}, with $l=1,2$, parametrize the operator $D_{\aux}$ while the auxiliary fields $\{h_{j}\}$ and $\{k_{l}\}$ are the components of the vectors in the Hilbert space $\mathcal{H}_{\aux}$. Moreover, we keep in mind the possibility of encountering a real spectral triple with mixed KO-dimension.

\begin{remark}
\label{Remark: linear fermionic action}
Since the action $S_{\aux}$ is not bilinear in the fields, it can not be expected to directly agree with a usual fermionic action. For this reason, we will slightly adapt the definition of a fermion action associated to a real spectral triple. 
\end{remark}

 \subsection*{The Hilbert space $\H_{\aux}$}
The Hilbert space describes the field content of the action $S_{\aux}$. So we have
$$\mathcal{H}_{\aux}= \mathcal{H}_{h} \oplus \mathcal{H}_{k}  :=M_2(\C) \oplus \C^2 .
$$
Again, we take the Pauli matrices  (\ref{basis}) as an orthonormal basis for $M_{2}(\mathbb{C})$, so that $\mathcal{H}_{\aux} \cong \mathbb{C}^{6}.$ We also identify the following subspace: 
$$\mathcal{H}_{\aux, f} = u(2) \oplus i [u(1) \oplus u(1)],$$
and write elements $\chi \in \mathcal{H}_{\aux, f}$ suggestively as
 $$\chi = [i h_{1}, \mbox{ }i h_{2}, \mbox{ }i h_{3}, \mbox{ }i h_{4}, \mbox{ }k_{1}, \mbox{ }k_{2}]^{T},$$
where $h_{j}$ and $k_{l}$, ($j= 1, \dots, 4$, $l= 1,2$) are real variables. %The notation has been chosen on purpose: in fact, the real variables $h_{j}$ and $k_{l}$ which represent the real components of a vector in $\mathcal{H}_{aux, f}$ will play the role of the auxiliary fields in our matrix model.

\subsection*{The real structure $J_{\aux}$}
The anti-linear isometry $J_{\aux}$ is defined similarly as in the BV spectral triple: indeed, $J_{\aux}:  \mathcal{H}_{h}\oplus \mathcal{H}_{k} \rightarrow \mathcal{H}_{h} \oplus \mathcal{H}_{k}$, with 
 $$J_{\aux}(V, v):= (i \cdot V^{*}, i \cdot \overline{v}).$$
% $$\begin{array}{cccl}
%J_{\aux}: & \mathcal{H}_{h}\oplus \mathcal{H}_{k} & \longrightarrow & \mathcal{H}_{h} \oplus \mathcal{H}_{k}\\
%[1ex] & (V, v) & \mapsto & J_{\aux}(V, v):= (i \cdot V^{*}, i \cdot \overline{v}) .
%   \end{array}
% $$
%where $R \in \mathcal{H}_{h}= M_{2}(\mathbb{C})$, with $R^{*}$ its adjoint matrix, and $v \in \mathbb{C}^{2}$, with $\overline{v}$ its complex conjugate, considered componentwise. 
 
\subsection*{The operator $D_{\aux}$}
In the basis of $M_2(\C)$ given by the Pauli matrices we define the operator $D_{\aux}$ as
$$
D_{\aux} = D_{\diag} + D_{\off} \quad \text{with} \quad D_\diag = \begin{pmatrix} P & 0 \\ 0 & 0 \end{pmatrix} , \qquad D_\off = \begin{pmatrix} 0 & Q^* \\ Q & 0  \end{pmatrix}
$$
where, again in evocative notation in terms of the anti-fields $A^{*}_{l}$  and $B^{*}_{j}$, we define
{\small{\begin{gather*}
P  = \frac12 \begin{pmatrix} 
+B^{*}_{1} + B^{*}_{2} + B^{*}_{3} & 0 & 0 & +B^{*}_{1} - B^{*}_{2} - B^{*}_{3}\\
0 & +B^{*}_{1} + B^{*}_{2} + B^{*}_{3}  & 0 & -B^{*}_{1} + B^{*}_{2} - B^{*}_{3}\\
0 & 0 & +B^{*}_{1} + B^{*}_{2} + B^{*}_{3} & -B^{*}_{1} - B^{*}_{2} + B^{*}_{3}\\
+B^{*}_{1} - B^{*}_{2} - B^{*}_{3} & -B^{*}_{1} + B^{*}_{2} - B^{*}_{3} & -B^{*}_{1} - B^{*}_{2} + B^{*}_{3} & +B^{*}_{1} + B^{*}_{2} + B^{*}_{3}
\end{pmatrix}
\intertext{and} 
Q = -\frac i 3 \begin{pmatrix}  A_1^* & A_1^* & A_1^* & 0 \\A_2^* & A_2^* & A_2^* & 0 \end{pmatrix}
\end{gather*}
}}
\subsection*{The algebra $\A_{\aux}$}
Also in this case, we take the algebra $\A_{\aux}$ to be the largest unital subalgebra of $\mathcal{L} (\H_\aux)$ that completes the triple $\H_\aux$, $D_\aux$, and $J_\aux$ to a real  spectral triple (with mixed KO-dimension). 
\begin{lma}
\label{algebra per BV auxiliary sp tr} 
Let $\H_\aux, J_\aux$ and $D_\aux$ be as defined above. Then the maximal unital subalgebra $\tilde \A$ of $\mathcal{L}(\H_\aux)$ on which the commutation rule and first-order condition are fulfilled, that is,  which satisfies
$$
[a,J_\aux b^* J^{-1}_\aux]= 0, \qquad [[\H_\aux,a],J_\aux b^* J^{-1}_\aux]= 0; \qquad (a,b \in \tilde \A),
$$
is $\tilde \A = \C$. 
\end{lma}

\begin{proof}
The fact that $\H_\aux$ should be a $\tilde \A$-bimodule already restricts $\tilde \A$ to be a subalgebra of 
$$
 M_{2}(\mathbb{C}) \oplus \mathbb{C}\oplus \C.
$$
acting by (block diagonal) matrix multiplication on $\H_\aux = M_2(\C) \oplus \C \oplus \C$. A straighforward computation of the double commutator entering in the first-order condition then selects the diagonal subalgebra $\tilde A = \C$. 
\end{proof}
We will write $\A_\aux = \C$ and notice the intriguing agreement between the triviality of the algebra with the triviality of the trivial pairs of auxiliary fields. 

\begin{prop}
\label{sp triple bosonica e fermionica per sp triple BV auxiliary}
For $\mathcal{A}_{\aux}$, $\mathcal{H}_{\aux}$, $D_{\aux}$ and $J_{\aux}$ as previously defined, it holds that
\begin{enumerate}[(i)]
 \item $(\mathcal{A}_{\aux}, \mathcal{H}_{\aux}, D_{\diag}, J_{\aux})$ is a real spectral triple of KO-dimension $7$;
 \item $(\mathcal{A}_{\aux}, \mathcal{H}_{\aux}, D_{\off}, J_{\aux})$ is a real spectral triple of KO-dimension $1$.
\end{enumerate}
\end{prop}

\begin{proof}
Because $\mathcal{A}_{\aux} = \mathbb{C}$, this follows at once from noticing that
$$J_{\aux}D_{\diag} = + D_{\diag}J_{\aux} \quad \quad \mbox{ and } \quad \quad J_{\aux}D_{\off} = - D_{\off}J_{\aux},$$
which can be readily checked.
\end{proof}

The last ingredient to analyze is the fermionic action. As already noticed in Remark \ref{Remark: linear fermionic action}, we need to introduce a linear notion of fermionic action. More precisely, the fermionic action corresponding to the operators $D_{\diag}$ and $D_{\off}$ will be defined using two linear forms $\mathcal{L}_{D_{\diag}}$, $\mathcal{L}_{D_{\off}}$ instead of a bilinear form $\mathfrak{A}$, as was done for the BV action. This is a consequence of the fact that the auxiliary action $S_{\aux}$ is only linear (rather than quadratic) in the fields. We state the following general, but straightforward result without proof. 

\begin{prop}
\label{linear form for BV auxiliary spectral triple}
Let $(\A,\H,D,J)$ be a real spectral triple (possibly with mixed KO-dimension) and fix a vector $v \in \H$. Then the expression $$\mathcal{L}_{D}(\chi) = \frac 12 \left(\langle Jv, D\chi \rangle + \langle J\chi, Dv\rangle \right), \qquad (\chi \in \H)$$ defines a linear form on $\mathcal{H}$. 
\end{prop}
In our case of interest, we fix the vector $v$ to be 
$$
v = \underline{1} = \begin{pmatrix} 1 & 1 & 1 & 1 & 1 & 1 \end{pmatrix}^T \in \C^6 \equiv \H_\aux .
$$

\begin{theorem}
\label{teorema BV auxiliary spectral triple}
$(\mathcal{A}_{\aux}, \mathcal{H}_{\aux}, D_{\aux}, J_{\aux})$ defines a real spectral triple with mixed KO-dimension. Moreover, the fermionic action defined by the linear form $\mathcal{L}_{D_{\aux}}$ coincides with the auxiliary action $S_{\aux}$:
$$S_{\aux} = \frac{1}{2} \left(\langle J_{\aux}(\underline{1}), D_{\aux}(\chi) \rangle + \langle J_{\aux}(\chi), D_{\aux}(\underline{1})\rangle \right), \qquad \text{with} \quad \chi \in \H_{\aux,f}.$$
Here we interpret the variables that parametrize $D_{\aux}$ and the vector $\chi \in \H_\aux$ as follows: 
\begin{itemize}
\item $B_j^*, h_l$ are real variables;
\item $A_l^*, k_l$ are Grassmannian variables. 
\end{itemize}
\end{theorem}

\begin{proof}
For a generic vector $\chi$ in $\mathcal{H}_{\aux, f}$, $\chi = [ i h_{1}, i h_{2}, i h_{3}, i h_{4}, k_{1}, k_{2}]^{T},$ one computes that
\begin{align*}
\langle J_{\aux}(\underline{1}), D_{\aux}(\chi) \rangle &= \sum_{j,l} B^{*}_{j}h_{j} + A^{*}_{l}k_{l} - \tfrac{i}{3}A^{*}_{l}h_{j}\\
 \langle J_{\aux}(\chi), D_{\aux}(\underline{1})\rangle &=  \sum_{j,l} h_{j} B^{*}_{j} -  k_{l}A^{*}_{l} + \tfrac{i}{3}h_{j}A^{*}_{l} 
\end{align*}
with $j=1, 2, 3$, $l=1, 2$. It then follows that 
$$
\frac{1}{2} \left[\langle J_{\aux}(\underline{1}), D_{\aux}(\chi) \rangle + \langle J_{\aux}(\chi), D_{\aux}(\underline{1})\rangle \right] = \sum_{j=1}^{3} B^{*}_{j}h_{j} + \sum_{l=1}^{2} A^{*}_{l}k_{l},
 $$
\noindent
whose right-hand side coincides with $S_{\aux}$ as defined in (\ref{eq: azione triv}).  
\end{proof}
As expected, the BV auxiliary spectral triple has a similar structure to the one already found for the BV spectral triple: 
\begin{itemize}
 \item  The anti-fields $B^{*}_{j}$ and $A^{*}_{l}$ appear as entries of the operator $D_\aux$ while the fields $h_{j} $ and $k_{l}$ are the components of the vectors in subspace $\H_{\aux,f}$. 
 \item The parities of the fields and anti-fields are a consequence of the structure of the real spectral triple, except perhaps for the parity of the anti-fields $ B^{*}_{l}$.
 \end{itemize}

\section{A possible approach to BV spectral triples}
\label{Section: Outlook}
\noindent
The procedure presented in this paper allows to describe the BV construction of a given gauge theory in the setting of noncommutative geometry. Even though we have restricted ourselves to the case of a $U(2)$-gauge invariant matrix model, our results suggest a possible way on how to proceed in a more general setting. Indeed, let $(X_{0}, S_{0})$ be the gauge theory derived from a finite spectral triple $(\mathcal{A}, \mathcal{H}, D)$ along the lines of Section \ref{Section: The noncommutative geometry setting}. Then the BV formalism gives rise to an extended theory $(\widetilde{X}, \widetilde{S})$ which one tries to capture by a BV spectral triple
$$(\mathcal{A}_{\BV}, \mathcal{H}_{\BV}, D_{\BV}, J_{\BV}).$$
The properties that this spectral triple should satisfy are 
\begin{enumerate}
\item The algebra $\A_\BV$ coincides with $\A$;
\item The Hilbert space $\H_\BV$ is spanned by the gauge fields and all ghost fields;
\item The real structure selects the hermitian variables in $\H_\BV$.
\end{enumerate}
Of course, the main challenge is now to find the form of the operator $D_\BV$ in terms of the anti-fields for which the fermionic action coincides with the BV action functional. One of the problems to overcome here is that the BV action might have terms of order higher than 2 in the ghost fields, requiring the introduction of some sort of {\em multilinear fermionic action}. A first analysis of this is in progress for $U(n)$-matrix models with $n>2$.

\bibliographystyle{plain}

\end{document}